\documentclass[journal,12pt,onecolumn,draftcls]{IEEEtran}
\usepackage{float}
\usepackage{comment}
\usepackage[table]{xcolor}
\usepackage{cite}      
\usepackage{graphicx}  
\usepackage{psfrag}    
\usepackage{placeins}
\usepackage{lineno}
\usepackage{url}
\usepackage{subcaption}
\usepackage{caption}
\usepackage{times}
\usepackage{textcomp}
\usepackage{amsmath} 
\usepackage{amsfonts,amsthm,amssymb,mathtools}
\usepackage{balance}
\usepackage[normalem]{ulem}
\usepackage{xcolor}

\usepackage{cuted}
\usepackage{amsmath,bm}
\usepackage{hyperref}
\usepackage{cleveref}

\DeclareMathOperator{\E}{\mathbb{E}}
\DeclareMathOperator{\Pro}{\prod_{c \in \Phi \cap \mathbf{A_c}(\kappa_m)}}
\DeclareMathOperator{\rcsm}{\sigma_{m_{avg}}}
\DeclareMathOperator{\rcsc}{\sigma_{c_{avg}}}
\DeclareMathOperator{\Eii}{\underset{\sigma_c,c}{\E}}
\DeclareMathOperator{\Ei}{\underset{\sigma_c}{\E}}
\DeclareMathOperator{\Su}{\sum_{c \in \Phi \cap \mathbf{A_c}(\kappa_m)}}

\DeclareMathOperator{\hxt}{\mathcal{H}(\kappa_m)}
\DeclareMathOperator{\hxc}{\mathcal{H}(\kappa_c)}

\DeclareMathOperator{\Pdc}{\mathcal{P}_{DC}^{Bi}}
\DeclareMathOperator{\PdcPhi}{\mathcal{P}_{DC\Phi}^{Bi}}
\DeclareMathOperator{\thetam}{\mathbf{\theta_m}}
\DeclareMathOperator{\sigmam}{\mathbf{\sigma_m}}
\DeclareMathOperator{\sigmac}{\mathbf{\sigma_c}}
\DeclareMathOperator{\thetac}{\mathbf{\theta_c}}
\DeclareMathOperator{\longvar}{\rcsm P_{tx}G_0B_0\epsilon H_0}
\DeclareMathOperator{\onetheta}{\Delta \theta_{tx}}

\newtheorem{theorem}{Theorem}[]
\newtheorem{corollary}{Corollary}[theorem]

\usepackage{multicol}
\usepackage{xr-hyper}

\begin{document}
%\linenumbers
\title{Optimization of Network Throughput of Joint Radar Communication System Using Stochastic Geometry}
\author{Shobha~Sundar~Ram,~\IEEEmembership{Senior Member,~IEEE} and Gourab~Ghatak~\IEEEmembership{Member,~IEEE}
\thanks{The authors are with Indraprastha Institute of Information Technology Delhi (email:  shobha@iiitd.ac.in, gourab.ghatak@).iiitd.ac.in}}
% make the title area
\maketitle
\begin{abstract}
Recently joint radar communication (JRC) systems have gained considerable interest for several applications such as vehicular communications, indoor localization and activity recognition, covert military communications, and satellite based remote sensing. In these frameworks, bistatic/passive radar deployments with directional beams explore the angular search space and identify mobile users/radar targets. Subsequently, directional communication links are established with these mobile users. Consequently, JRC parameters such as the time trade-off between the radar exploration and communication service tasks have direct implications on the network throughput. Using tools from stochastic geometry (SG), we derive several system design and planning insights for deploying such networks and demonstrate how efficient radar detection can augment the communication throughput in a JRC system. Specifically, we provide a generalized analytical framework to maximize the network throughput by optimizing JRC parameters such as the exploration/exploitation duty cycle, the radar bandwidth, the transmit power and the pulse repetition interval. The analysis is further extended to monostatic radar conditions, which is a special case in our framework. The theoretical results are experimentally validated through Monte Carlo simulations. Our analysis highlights that for a larger bistatic range, a lower operating bandwidth and a higher duty cycle must be employed to maximize the network throughput. Furthermore, we demonstrate how a reduced success in radar detection due to higher clutter density deteriorates the overall network throughput. Finally, we show a peak reliability of 70\% of the JRC link metrics for a single bistatic transceiver configuration. 
%%% Leave the Abstract empty if your article does not require one, please see the Summary Table for full details.
\end{abstract}
\providecommand{\keywords}[1]{\textbf{\emph{Keywords--}}#1}

\IEEEpeerreviewmaketitle

\section{Introduction}
Over the last decade, joint radar communication (JRC) frameworks are being researched and developed for numerous applications at microwave and millimeter wave (mmWave) frequencies~\cite{liu2020joint}. Through the integration of sensing and communication functionalities on a common platform, JRC based connected systems offer the advantages of increased spectral efficiency through shared spectrum and reduced hardware costs. The most common applications are 
WiFi/WLAN based indoor detection of humans  \cite{falcone2012potentialities,storrer2021indoor,tan2016awireless,li2020passive,alloulah2019future,yildirim2021super}, radar enhanced vehicular communications \cite{ali2020passive,kumari2017ieee,dokhanchi2019mmwave,duggal2020doppler}, covert communications supported by radar based localization \cite{kellett2019random,hu2019optimal} and radar remote sensing based on global navigation satellite systems (GNSS) \cite{zavorotny2014tutorial}. All of these systems consist of a dual functional (radar-communication) transmitter and either a standalone or integrated radar/communications receiver. When the radar receiver is not co-located with the transmitter, the system constitutes a passive/bistatic radar framework. This is the most common scenario in sub-6GHz indoor localization systems where the WiFi access point is an opportunistic illuminator and humans activities are sensed for intrusion detection, surveillance, or assisted living. The bistatic scenario is also encountered in GNSS based remote sensing where the ground reflected satellite signals are analyzed, at a passive radar receiver, to estimate land and water surface properties \cite{zavorotny2014tutorial}. JRC based systems are also being researched for next generation intelligent transportation services where one of the main objectives is to share environment information for collision avoidance, and pedestrian detection eventually leading to autonomous driving. MmWave communication protocols such as IEEE 802.11ad/ay characterized by high bandwidths and low latency have been identified for vehicular-to-everything (V2X) communications \cite{nitsche2014ieee,zhou2018ieee}. However, due to the severe propagation loss at mmWave carrier frequencies, they are meant to operate in short range line-of-sight (LOS) conditions with highly directional beams realized through digital beamforming. In high mobility environments, beam training will result in considerable overhead and significant deterioration of latency. Hence, the integration of the radar functionality within the existing millimeter wave communication frameworks is being explored for rapid beam alignment \cite{kumari2017ieee,dokhanchi2019mmwave,duggal2019micro,grossi2021opportunistic}. The wide bandwidth supported by the mmWave signals along with the channel estimation capabilities within the packet preamble are uniquely suited for radar remote sensing operations. A preliminary work on the detection metrics of a bistatic radar was presented in \cite{ram2022Estimation}. In this work, we consider a generalized passive/bistatic radar framework that can be used to model the JRC application scenarios described above and analyze the communication network throughput performance as a function of radar detection metrics. The monostatic radar scenario is considered as a limiting case of the bistatic radar and the corresponding results are obtained as a corollary.

In all of the applications discussed above, the implementation of the dual functional systems involves the optimization of several hardware and software design parameters - such as antennas, transmit waveform and signal processing algorithms - for enhanced radar detection performance without deterioration in the communication metrics \cite{hassanien2016signaling,mishra2019toward,ma2021spatial}. 
In this work, we focus on the time resource management between the radar and communication functionalities for optimizing communication network throughput.
Prior works have tackled the time resource management for multi-functional radars \cite{miranda2007comparison}. In \cite{grossi2017two}, the radar dwell time was optimized for maximum target detection for a constant false alarm rate. In \cite{ghatak2021beamwidth}, the 
time resource management between the localization and communication functionalities was determined as a function of the density of base station deployment. During the radar/localization phase, the transmitter must scan the angular search space and determine the number and location of the mobile users. Then these users must be served during the remaining duration through directional/pencil beams. The exploration and service process must be repeated periodically due to the motion of the mobile user. Now, if the angular beamwidth of the search beams are very narrow, then they will take longer to cover the search space (for a fixed dwell time) and this will result in reduced communication service time. However, the radar link quality will be higher due to the improved gain and result in a larger number of targets being detected. Hence, the overall network throughput is a function of the explore/exploit time  management. In this paper, we use stochastic geometry (SG) based formulations to optimize the network throughput as a function of the explore/exploit duty cycle. 

SG tools were originally applied to communication problems in cellular networks, mmWave systems, and vehicular networks \cite{chiu2013stochastic,andrews2011tractable,bai2014coverage,thornburg2016performance,ghatak2018coverage}. In all of these scenarios, there is considerable variation in the strength and spatial distribution of the base stations.
More recently, they have been increasingly used in diverse radar scenarios to study the radar detection performance under interference and clutter conditions \cite{al2017stochastic,munari2018stochastic,ren2018performance,park2018analysis,fang2020stochastic}. These works have considered the significant diversity in the spatial distributions and density of radars. SG offers a mathematical framework to analyze performance metrics of 
spatial stochastic processes that approximate to Poisson point process distributions without the requirement of computationally expensive system simulation studies or laborious field measurements. Based on the mathematical analysis, insights are obtained of the impact of design parameters on system level performances. In our problem related to JRC, there can be considerable variation in the position of the dual functional base station transmitter, the radar receiver and the communication end users who are the primary radar targets. Additionally, the JRC will encounter reflections from undesired targets/clutter in the environment. We model the discrete clutter scatterers in the bistatic radar environment as a homogeneous Poisson point process (PPP) similar to \cite{chen2012integrated,ram2020estimating,ram2021optimization}. This generalized framework allows us to regard each specific JRC deployment, not as an individual case, but as a specific instance of an overall spatial stochastic process. Further, the target parameters such as the position and radar cross-section are also modelled as random variables. Using SG we quantify the mean number of mobile users that can be detected by the radar provided the statistics of the target and clutter conditions are known and subsequently determine the network throughput. Then we use the theorem to optimize system parameters such as the explore/exploit duty cycle, transmitted power, radar bandwidth and pulse repetition interval for maximum network throughput. Our results are validated through Monte Carlo simulations carried out in the short range bistatic radar framework.

Our paper is organized as follows. In the following section, we present the system model of the JRC with the bistatic radar framework and describe the explore/exploit time management scheme. In section \ref{sec:Theory}, we provide the theorem for deriving the network throughput as a function of the bistatic radar parameters. In section \ref{sec:Results}, we offer the key system parameter insights that are drawn from the theorem as well as the Monte Carlo simulation based experimental validation. Finally, we conclude the paper with a discussion on the strengths and limitations of the proposed analytical framework. 

\emph{Notation:} In this paper, all the random variables are indicated with bold font and constants and realizations of a random variable are indicated with regular font. 
\section{System Model}
\label{sec:SystemModel}
We consider a joint radar-communication (JRC) framework with a single base station (BS), multiple mobile users (MU) and a single passive radar receiver (RX) as shown in Fig.1a. The BS serves as a dual functional transmitter that supports both radar and communication functionalities in a time division manner as shown in Fig.1b. 
\begin{figure*}
    \centering
    \includegraphics[width=6in, height = 3in]{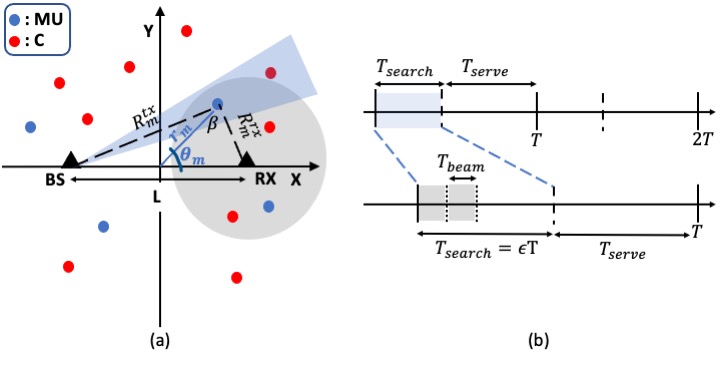}
    \caption{(a) Illustration of the joint radar-communication (JRC) scenario. The base station (BS) at $(\frac{L}{2},0)$ and indicated by a triangle is a dual functional transmitter that supports both radar and comm. functionalities with a directional and reconfigurable antenna system of $\onetheta$ beamwidth. An omnidirectional receiver (RX) at $(-\frac{L}{2},0)$ forms the bistatic/passive radar receiver. The channel consists of mobile users (MU) at $(r_m,\theta_m)$ at distances, $R^{tx}_m$ and $R^{rx}_m$, from BS and MU respectively indicated by blue dots; and undesirable clutter scatterers indicated by red dots. The bistatic radar angle is $\beta$. (b) Timing diagram of the JRC framework where each $T$ consists of $T_{search}=\epsilon T$ when the BS scans the angular search space for MU using $n_{beam}$ of $T_{beam}$ duration. During the remaining $T_{serve}$ duration, directional beam links are established between BS and MU based on the localization by the radar during $T_{search}$.}
    \label{fig:my_label}
\end{figure*}
During the $T_{search}$ interval, the BS serves as the radar transmitter or opportunistic illuminator and along with the RX, forms a bistatic radar whose objective is to localize the multiple MU in the presence of clutter/undesirable targets. During this interval, the BS transmits a uniform pulse stream of $\tau$ pulse width and $T_{PRI}$ pulse repetition interval, through a directional and reconfigurable antenna system with gain $G_{tx}$ and beamwidth $\onetheta$. The radar must scan the entire angular search space within $T_{search}$ to find the maximum number of MU. If the duration of an antenna beam is fixed at $T_{beam}$ (based on hardware parameters such as circuit switching speed for electronic scanning or Doppler frequency resolution requirements), then the number of beams that can be searched within $T_{search}$ is given by \par\noindent\small
\begin{align}
\label{eq:nbeam}
  n_{beam} = \frac{\Omega}{\onetheta} = \frac{T_{search}}{T_{beam}},  
\end{align}
where $\Omega$ is the angular search space. In our problem formulation, we set $\Omega = 2\pi$ to correspond to the entire azimuth angle extent. During the remaining duration of $T_{serve}$, directional communication links are assumed to be established between the BS and the detected MUs. \emph{Thus the beam alignment for communication during $T_{serve}$ is based on radar enabled localization during $T_{search}$.} 
Since the position of the MU does not remain fixed with time, the process of beam alignment is repeated for every $T = T_{search}+T_{serve}$ as shown in the figure. An important tuning parameter in the above JRC framework is the duty cycle $\epsilon = \frac{T_{search}}{T}$. From \eqref{eq:nbeam}, it is evident that $\onetheta = \frac{\Omega T}{\epsilon T_{search}} = \frac{1}{B_0\epsilon}$. Note that when the beams become broader, the gain of the radar links become poorer. As a result of the deterioration in the radar link metrics due to larger $\onetheta$, the detection performance becomes poorer and fewer MU ($\eta$) are likely to be detected in the search space. Thus $\eta$ is directly proportional to $\epsilon$. On the other hand, the network throughput ($\Upsilon$) of the system is defined as \par\noindent\small
\begin{align}
\label{eq:throughput}
\Upsilon = \eta(\epsilon)(1-\epsilon)D,
\end{align}
where $(1-\epsilon)$ is the duty cycle of the communication service time ($\frac{T_{search}}{T}$). Here, we assume that the communication resources such as spectrum are available to all the $\eta$ detected MU and all the MU are characterized by identical $D$.  
The objective of our work is to present a theorem to optimize the duty cycle $\epsilon$ for maximum $\Upsilon$ under the assumption that the noise, MU and clutter statistics are known and fixed during the radar processing time. These conditions are generally met for microwave or millimeter-wave systems \cite{billingsley2002low,ruoskanen2003millimeter}. The theoretical framework is derived for a generalized bistatic JRC framework where inferences for monostatic conditions are derived from limiting conditions.

Next, we discuss the bistatic radar geometry that we have considered based on the north-referenced system described in \cite{jackson86}. We assume that the BS is located in the Cartesian coordinates $(-\frac{L}{2},0)$ while the passive receiver, RX, is assumed to be omnidirectional and located at $(+\frac{L}{2},0)$. 
This is the most common framework in many modern passive radar deployments \cite{willis2005bistatic,davis2007advances,falcone2012potentialities}. High gain transmission links from the BS support high quality communication link metrics. The gain of  the passive RX antenna is intentionally kept low so that the common search space of the bistatic radar transmitter and receiver does not become too narrow which would then have to be supported by very time consuming and complicated beam scanning operations.  
The baseline length between the bistatic radar transmitter and receiver is $L$. The two-dimensional space is assumed to be populated by multiple scatterers - some MU ($m$) and the remaining discrete clutter ($c$) scatterers. In real world conditions, there can be significant variation in the number and spatial distribution of the point scatterers (both MU and clutter) in the radar channel. Further, the positions of scatterers are independent of each other. Consequently, we consider the distribution of both types of scatterers as independent Poisson point processes (PPP) - wherein each instance is assumed to be a realization of a spatial stochastic process. The number of the scatterers in each realization follows a Poisson distribution and the positions of these scatterers follow a uniform distribution. Some prior works where discrete scatterers have been modelled as a PPP are  \cite{chen2012integrated,ram2020estimating,ram2021optimization}. We assume that that the mean spatial densities of the MU and clutter scatterers are $\rho_m$ and $\rho_c$ respectively where $\rho_m << \rho_c$. 
The position of an MU/clutter scatterer is specified in polar coordinates $(r_i,\theta_i), i\in m,c$ where $r_i$ is the distance from the origin and $\theta_i$ is the angle with the positive $X$ axis.  The distance from BS and RX  are $R^{tx}_i$ and $R^{rx}_i$ respectively and the bistatic range ($\kappa_i$) is specified by the geometric mean of both the one-way propagation distances ($\kappa_i = \sqrt{R^{tx}_iR^{rx}_i}$). 
In bistatic radar geometry, the contours of constant $\kappa_i$ for a fixed $L$ are called  Cassini ovals \cite{willis2005bistatic}. Two regions are identified: the first is the \emph{cosite} region when $L \leq 2\kappa_i$ and the contours appear as concentric ovals for different $\kappa_i$; and the second is when $L > 2\kappa_i$ and the oval splits into two circles centered around BS and RX. In our work, we assume that cosite conditions prevail and that the bistatic angle at MU is $\beta$. Note that when $L$ is zero, $\beta = 0$ and the system becomes a monostatic radar scenario. Here, the Cassini ovals become concentric circles for different values of $R^{tx}_i = R^{rx}_i =\kappa_i$. 

In \cite{ram2022Estimation}, we presented a metric called the radar detection coverage probability ($\Pdc$) to indicate the likelihood of a radar target being detected by a bistatic radar based on the signal-to-clutter-and-noise ratio (SCNR). The metric is analogous to wireless detection coverage probability which is widely studied in communication systems to study the network coverage in wireless links \cite{andrews2011tractable}. The metric is preferred to other more conventional radar metrics such as probability of detection and probability of false alarm since it offers physics based insights into system performance and because of its tractable problem formulation. In this work, we use this metric to estimate the mean number of detected MU ($\eta$) as a function of $\epsilon$ and optimize the network throughput ($\Upsilon$). If the transmitted power from BS is $P_{tx}$ and the bistatic radar cross-section (RCS) of the MU, $\sigma_m$, is a random variable, then the received signal at RX, $S$, is given by the Friis radar range equation as \par\noindent\small
\begin{align}
\label{eq:TgtSignal}
    \textbf{S}(\kappa_m) = P_{tx}G_{tx}(\thetam)\mathbf{\sigma_m}\hxt,
\end{align}
where $\hxt$ is the two-way propagation factor. In line-of-sight (LOS) conditions this is 
\par\noindent\small
\begin{align}
\label{eq:hxt}
    \hxt = \frac{\lambda^2}{(4\pi)^3 (R^{tx}R^{rx})^2}= \frac{H_0}{\kappa_m^{4}},
\end{align}
where $\lambda$ is the wavelength of the radar. In the above expression, the gain of RX is assumed to be 1 since it is an omnidirectional antenna. We assume that the gain of the BS is uniform within the main lobe and is inversely proportional to the beam width: $G_{tx}=\frac{G_0}{\onetheta}$ where $G_0$ is the constant of proportionality that accounts for antenna inefficiencies including impedance mismatch, dielectric and conductor efficiencies. 
If we assume that the MU is within the mainlobe of the radar, then using \eqref{eq:nbeam}, equation \eqref{eq:TgtSignal} can be written as \par\noindent\small  
\begin{align}
\label{eq:TgtSignal2}
    \textbf{S}(\kappa_m) =  \frac{P_{tx}G_0\sigmam\hxt}{\onetheta} = P_{tx}G_0B_0\epsilon\sigmam\hxt.
\end{align}
In \eqref{eq:TgtSignal} and \eqref{eq:TgtSignal2}, we have assumed that only a single MU is within a radar resolution cell, $\mathbf{A_c}$, since multiple MUs are indistinguishable to the radar if they fall within the same cell for identical $\kappa_m$ and $\thetam$. Further, we assume that the $\sigmam$ follows the Swerling based distribution which models a radar target as an extended target with multiple scattering centers within a single resolution cell. The clutter returns, $\mathbf{C}$, at the radar receiver is given by \par\noindent\small
\begin{align}
\label{eq:CluttSignal}
    \mathbf{C}(\kappa_m) = \Su P_{tx}G_{tx}(\thetac)\mathbf{\sigma_c}\hxc.
\end{align}
In the above expression, we specifically only consider those clutter scatterers that fall within the same resolution cell, $\mathbf{A_c}$, as the MU. We use the generalized Weibull model \cite{sekine1990weibull} to describe the distribution of the RCS ($\sigmac$) of the clutter points. For a given noise of the radar receiver, $N_s = K_BT_sBW$ where $K_B, T_s$ and $BW$ are the Boltzmann constant, system noise temperature and bandwidth respectively, the signal to clutter and noise ratio is given by $\mathbf{SCNR}(\kappa_m) = \frac{\mathbf{S}(\kappa_m)}{\mathbf{C}(\kappa_m)+N_s}.$
\section{Estimation of Network Throughput of JRC}
\label{sec:Theory}
In this section we present the analytical framework to estimate the network throughput of the communication framework as a function of the explore/exploit duty cycle ($\epsilon$). We use the $\Pdc$ metric defined in \cite{ram2022Estimation} to estimate, $\eta$, the number of MU detected by the radar during the search interval $T_{search}=\epsilon T$ that will be subsequently served during $T_{serve}$.
\begin{theorem}
The network throughput ($\Upsilon$) for an explore/exploit duty cycle ($\epsilon$) for a bistatic/passive radar based JRC system is given by 
\begin{align}
\label{eq:theorem_part1}
\Upsilon = \Pdc \left( 2\pi \kappa_m - \frac{3 \pi L^2}{8\kappa_m}\right)\frac{\rho_m c\tau}{2\sqrt{1-\frac{L^2}{4\kappa_m^2}}}(1-\epsilon)D
\end{align}
where
\begin{align}
\label{eq:theorem_part2}
\Pdc =exp \left( \frac{-\gamma N_s \kappa_m^4}{\longvar}+ \frac{-\gamma\rho_c c\tau\kappa_m^2\rcsc}{B_0\epsilon(\kappa_m+\sqrt{\kappa_m^2-L^2})(\rcsm + \gamma\rcsc)}\right)
\end{align}
\end{theorem}
\begin{proof}
For an MU at bistatic range $\kappa_m$, the $\mathbf{SCNR}$ is a function of several random variables such as the MU cross-section, the position of MU, the number and spatial distribution of the discrete clutter scatterers and their RCS as shown below - \par\noindent\small
\begin{align}
\label{eq:SCNR1}
\mathbf{SCNR}(\kappa_m) 
= \frac{P_{tx}G_0B_0\epsilon\mathbf{\sigma_m}\hxt}{\Su P_{tx}G_0B_0\epsilon\mathbf{\sigma_c}\hxc+N_s} = \frac{\mathbf{\sigma_m}}{\Su\frac{\mathbf{\sigma_c}\hxc}{\hxt}+\frac{N_s}{P_{tx}G_0B_0\epsilon\hxt}}.
\end{align}
We define the bistatic radar detection coverage probability ($\Pdc$) as the probability that the SCNR is above a predefined threshold, $\gamma$. Therefore, \par\noindent\small
\begin{align}
\label{eq:Pdc1}
\Pdc = \mathcal{P}(\mathbf{SCNR}(\kappa_m)\geq \gamma)=\mathcal{P}\left(\mathbf{\sigma_m}\geq \Su\frac{\gamma\mathbf{\sigma_c}\kappa_m^4}{\kappa_c^4}+\frac{\gamma N_s \kappa_m^4}{P_{tx}G_0B_0\epsilon H_0}\right).
\end{align}
The bistatic RCS, $\sigmam$, has been shown to demonstrate similar statistics as monostatic RCS \cite{skolnik1961analysis}. In this work, we consider the MU to have Swerling-1 characteristics, which corresponds to mobile users such as vehicles and humans \cite{raynal2011radar,Raynal2011RCS}, as shown below \par\noindent\small
\begin{align}
\label{eq:TargetRCS}
  \mathcal{P}(\sigma_m) = \frac{1}{\rcsm}exp \left(\frac{-\sigma_m}{\rcsm} \right),
\end{align}
where, $\rcsm$ is the average radar cross-section.
Hence, \eqref{eq:Pdc1} can be expanded to \par\noindent\small
\begin{align}
\label{eq:Pdc2}
\Pdc =exp \left(\Su\frac{-\gamma\mathbf{\sigma_c}}{\rcsm}-\frac{\gamma N_s\kappa_m^4}{\longvar} \right) =exp \left( \frac{-\gamma N_s \kappa_m^4}{\longvar}\right) I(\kappa_m).
\end{align}
In the above expression, $\Pdc$ consists of two terms. The first term consists entirely of constants and demonstrates the radar detection performance as a function of the signal-to-noise ratio (SNR). The second term, $I(\kappa_m)$, shows the effect of the signal-to-clutter ratio (SCR). Since, we are specifically considering the clutter points that fall within the same resolution cell, $\mathbf{A_c}$, as the MU we can assume that $\hxc \approx \hxt$ in \eqref{eq:Pdc1}. We provide further insights into this path loss approximation in our later sections. Finally, the exponent of sum of terms can be written as a product of exponents. Hence, $I(\kappa_m)$ is \par\noindent\small
\begin{align}
\label{eq:I1}
I(\kappa_m) = \Eii\left[\Pro exp \left(\frac{-\gamma\mathbf{\sigma_c}}{\rcsm} \right)\right].
\end{align}
The probability generating functional (PGFL) of a homogeneous PPP \cite{haenggi2012stochastic} based on stochastic geometry formulations is given as \par\noindent\small
\begin{align}
\label{eq:I2}
I =exp\left(-\Eii\left[\iint_{\mathbf{r_c,\phi_c}} \rho_c\left(1- exp\left(\frac{-\gamma\mathbf{\sigma_c}}{\rcsm}\right)\right) d(\vec{x}_c) \right]\right),
\end{align}
where $\rho_c$ is the mean spatial density of the clutter scatterers. The integral specifically considers the clutter scatterers that fall within the same resolution cell as the MU. Bistatic radar literature identifies three types of resolution cells - the range resolution cell, the beamwidth resolution cell and the Doppler resolution cell. In our study, the main objective of the radar is to perform range-azimuth based localization. Hence, we consider the range resolution cell, which based on \cite{willis2005bistatic}, corresponds to \par\noindent\small
\begin{align}
\label{eq:RangeResCell1}
    \mathbf{A_{c}}(\kappa_m) = \frac{c\tau R^{tx}(\thetam) \Delta \theta_{tx}}{2\cos^2 (\beta(\thetam)/2)} = \frac{c\tau R^{tx}(\thetam)}{B_0\epsilon(1+\sqrt{1-\sin^2\beta(\thetam)})},
\end{align} 
for a pulse width of $\tau$. In the above expression, note that the size of $\mathbf{A_c}$ varies as a function of constant $\kappa_m$ and the random variable $\thetam$. Prior studies show that $\sin\beta$ takes on the value of $\sin\beta_{max}$ with a very high probability when $R^{tx}_m \approx \kappa_m$ \cite{ram2022Estimation}. Based on bistatic geometry $\sin\beta_{max}=\sqrt{\frac{L^2}{\kappa_m^2}-\frac{L^4}{\kappa_m^4}}\approx \frac{L}{\kappa_m}$ when $\kappa_m>L$. Therefore, \eqref{eq:RangeResCell1} reduces to
\begin{align}
\label{eq:RangeResCell2}
    A_{c} \approx \frac{c\tau \kappa_m^2}{B_0\epsilon(\kappa_m+\sqrt{\kappa_m^2-L^2})}
\end{align}
If we assume that the clutter statistics are uniform within $A_c$, then the integral in \eqref{eq:I2} can be further reduced to
\par\noindent\small
\begin{align}
\label{eq:I3}
I=exp\left(-\Ei\left[\left(1- exp\left(\frac{-\gamma\mathbf{\sigma_c}}{\rcsm}\right)\right) \rho_c A_c\right]\right)
= exp\left(-\Ei\left[\left(1- exp\left(\frac{-\gamma\mathbf{\sigma_c}}{\rcsm}\right)\right) \frac{\rho_c c\tau\kappa_m^2}{B_0\epsilon(\kappa_m+\sqrt{\kappa_m^2-L^2})}\right]\right)
\end{align}
If we define $J(\kappa_m) = \frac{\rho_c c\tau\kappa_m^2}{B_0\epsilon(\kappa_m+\sqrt{\kappa_m^2-L^2})}$ as a constant independent of $\sigma_c$, then it can be pulled out of the integral for computing the expectation as shown below - \par\noindent\small
\begin{align}
\label{eq:I4}
I(\kappa_m)
=exp\left(-J(\kappa_m)\int_0^{\infty}\left(1- exp\left(\frac{-\gamma\mathbf{\sigma_c}}{\rcsm}\right)\right)\mathcal{P}(\sigma_c)d\sigma_c \right).
\end{align}
In our work, we specifically consider the contributions from discrete/point clutter responses that arise from direct and multipath reflections from the surrounding environment. We model the radar cross-section of these scatterers using the generalized Weibull model shown in \par\noindent\small
\begin{align}
\label{eq:ClutterRCS}
 \mathcal{P}(\sigma_c) = \frac{\alpha}{\rcsc}\left(\frac{\sigma_c}{\rcsc}\right)^{\alpha-1}\exp\left(-\left(\frac{\sigma_c}{\rcsc}\right)^{\alpha} \right),
\end{align}
where $\rcsc$ is the average bistatic radar cross-section and $\alpha$ is the corresponding shape parameter. The Weibull distribution has been widely used to model clutter due to its tractable formulation and its adaptability to different environment conditions \cite{sekine1990weibull}.
When the scenario is characterized by few dominant scatterers, $\alpha$ is near one and corresponds to the exponential distribution. On the other hand, when there are multiple scatterers of similar strengths, then $\alpha$ tends to two which corresponds to the Rayleigh distribution. The actual value of $\alpha$ in any real world scenario is determined through empirical studies. $I(\kappa_m)$ in \eqref{eq:I4} can be numerically evaluated for any value of $\alpha$. But for $\alpha=1$, the expression becomes
\begin{align}
\label{eq:I5}
I(\kappa_m) = exp\left(-\frac{\gamma J(\kappa_m) \rcsc}{\rcsm + \gamma\rcsc}\right).
\end{align}
Substituting \eqref{eq:I5} in \eqref{eq:Pdc2}, we obtain
\begin{align}
\label{eq:Pdc3}
\Pdc =exp \left( \frac{-\gamma N_s \kappa_m^4}{\longvar}+ \frac{-\gamma\rho_c c\tau\kappa_m^2\rcsc}{B_0\epsilon(\kappa_m+\sqrt{\kappa_m^2-L^2})(\rcsm + \gamma\rcsc)}\right) .
\end{align}
The above expression shows the probability that a MU at $\kappa_m$ is detected by the bistatic radar based on its SCNR. If we assume a uniform spatial distribution, $\rho_m$, of the MU in Cartesian space, then the mean number of MU that can be detected within the total radar field-of-view at $\kappa_m$ bistatic range from the radar will be given by
\begin{align}
\label{eq:eta2}
    \eta = \Pdc(\kappa_m)\rho_m \mathcal{C}(\kappa_m) \delta r,
\end{align}
where $\mathcal{C}(\kappa_m)$ is the circumference of a Cassini oval and $\delta r = \frac{c \tau}{2 \cos (\beta/2)}$ is the range resolution of the radar.
The parametric equation for the Cassini oval is given in
\begin{align}
\label{eq:polarCoord}
\left(r_m^2 + \frac{L^2}{4}\right)^2 -r_m^2L^2\cos^2\theta_m = \kappa_m^4.
\end{align}
Hence, the circumference $\mathcal{C}(\kappa_m)$ can be computed from
\begin{align}
\label{eq:CircumferenceCassiniOval}
    \mathcal{C}(\kappa_m) = \int_0^{2\pi} r_m(\theta_m) d\theta_m =  \frac{L}{2}\int_0^{2\pi} \left[\cos 2\theta_m \pm \left(\frac{16\kappa_m^4}{L^4} -\sin^2\theta_m\right)^{1/2} \right]^{1/2}d\theta _m\approx
    2\pi \kappa_m - \frac{3 \pi L^2}{8\kappa_m}.
\end{align}
When $\kappa_m >L$, the estimation of \eqref{eq:CircumferenceCassiniOval} can be approximated to the expression shown above. Note that for very large values of $\kappa_m>>L$, the scenario approaches monostatic conditions. Here, the oval approximates to a circle of circumference $2\pi\kappa_m$. Also, as mentioned before $\beta$ can be approximated to $\beta_{max}$. Hence $\cos(\beta_{max}/2) \approx \sqrt{1-\frac{L}{4\kappa_m^2}}$. Therefore, the mean number of detected MU is 
\begin{align}
\eta = \Pdc \left( 2\pi \kappa_m - \frac{3 \pi L^2}{8\kappa_m}\right)\frac{\rho_m c\tau }{2\sqrt{1-\frac{L^2}{4\kappa_m^2}}},
\end{align}
and the resulting network throughput for the communication links that are set up with detected MUs is
\begin{align}
\label{eq:Throughput2}
\Upsilon = \Pdc \left( 2\pi \kappa_m - \frac{3 \pi L^2}{8\kappa_m}\right)\frac{\rho_m c\tau}{2\sqrt{1-\frac{L^2}{4\kappa_m^2}}}(1-\epsilon)D.
\end{align}
\end{proof}
%\end{comment}
\section{Optimization JRC System Parameters for Maximization of Network Throughput}
\label{sec:Results}
In this section, we discuss the corollaries from the theorem presented in the previous section. Based on these inferences, we present how JRC parameters such as $\epsilon, \tau, \onetheta$ and $T_{PRI}$ can be optimized for maximum throughput. 
The results presented in this section are experimentally validated using Monte Carlo simulations. For the simulations, we assume that the bistatic radar transmitter (BS) and receiver (RX) are located at $(\pm\frac{L}{2},0)$ respectively as shown in Fig.\ref{fig:MonteCarloSetup}.
\begin{figure*}
    \centering
    \includegraphics[width=\linewidth]{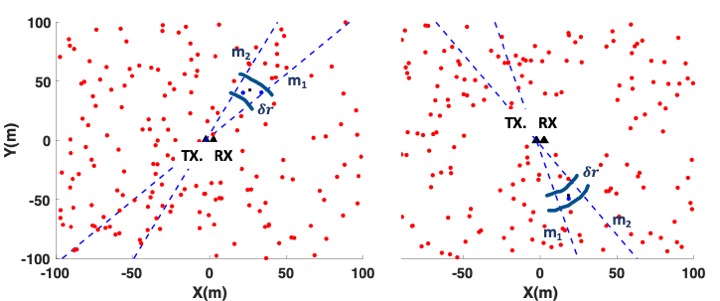}
    \caption{Two realizations of Monte Carlo simulations with bistatic radar transmitter (BS) and receiver (RX) indicated by triangles. The BS is characterized by narrow beam indicated by dashed blue lines with slopes $m_1$ and $m_2$ while RX is omnidirectional. Target is indicated by black dot while clutter scatterers \emph{inside} and \emph{outside} the radar resolution cell are indicated by blue and red dots respectively. }
    \label{fig:MonteCarloSetup}
\end{figure*}
We consider a $[200m \times 200m]$ region of interest.
Radar, MU and clutter parameters such as $P_{tx},L,\onetheta,N_s.\rcsm,\kappa,\rcsc$ and $\rho_c$ are kept fixed and summarized in Table.\ref{tab:my_params}. In each realization of the Monte Carlo simulation, the MU's polar coordinate position, $\theta_m$ is drawn from a uniform distribution from $[0,2\pi)$ and $r_m$ is computed for a fixed $\kappa_m$. The RCS of the MU is drawn from the exponential distribution corresponding to the Swerling-1 model. The mean number of  discrete clutter scatterers is equal to $\rho_c$ times the area of the region of interest. The number of clutter scatterers are different for each realization and drawn from a Poisson distribution. The positions of the clutter scatterers are based on a uniform distribution in the two-dimensional Cartesian space while the RCS of each discrete scatterer is drawn from the Weibull model. We compute the SCNR  based on the returns from the MU and the clutter scatterers estimated with the Friis bistatic radar range equation. Note that we only consider those point clutter that fall within the BS mainlobe and within $\delta r$ proximity of the two-way distance of the radar and MU. In other words, they must lie within the radar range limited resolution cell. 
To do so, we compute the slope of the line joining the scatterer and BS ($m_0$). Then we compute $m_1 = m_0+\tan(\onetheta/2)$ and $m_2=m_0-\tan(\onetheta/2)$ based on the radar BS beamwidth ($\onetheta$). The scatterer is within the radar beamwidth provided the product of the differences $(m_1-m_0)$ and $(m_2-m_0)$ is negative. Then we check if the absolute difference of the two-way path lengths of MU $(R^{tx}_m + R^{rx}_m)$ and point clutter ($R^{tx}_c+R^{rx}_c$) is within the range resolution $\delta r$. If the resulting SCNR is above the predefined threshold $\gamma$, then we assume that the target is detected. The results over a large number of realizations are used to compute the $\Pdc$. Note that the Monte Carlo simulations are useful to test some key assumptions made in SG based analysis such as the path loss approximation of the point clutter within the radar range limited resolution cell to the path loss of the MU. 
\subsection{Explore/Exploit Duty Cycle ($\epsilon$)}
\label{subsec:OptEpsilon}
In the JRC framework, a key parameter is $\epsilon = \frac{T_{search}}{T}$, the duty cycle, of the system. When $\epsilon$ is high, there is longer time for radar localization ($T_{search}$) but less time for communication service ($T_{serve}$) and vice versa. As a result, the radar beams can be narrow while scanning the angular search space. This results in weaker detection performance due to poorer gain. The theorem \eqref{eq:theorem_part1} shows the dependence of throughput $\Upsilon$ on $\epsilon$ which can be written as
\begin{align}
    \Upsilon(\epsilon) = A_0e^{-a/\epsilon}(1-\epsilon),
\end{align}
where 
\begin{align}
\label{eq:aconst}
    a = \frac{-\gamma N_s \kappa_m^4}{\rcsm P_{tx}G_0B_0H_0}+ \frac{-\gamma\rho_c c\tau\kappa_m^2\rcsc}{B_0(\kappa_m+\sqrt{\kappa_m^2-L^2})(\rcsm + \gamma\rcsc)}
\end{align}
and 
\begin{align}
\label{eq:bconst}
    A_0 = \left( 2\pi \kappa_m - \frac{3 \pi L^2}{8\kappa_m}\right)\frac{\rho_m c\tau D}{2\sqrt{1-\frac{L^2}{4\kappa_m^2}}}.
\end{align}
We find the optimized $\tilde{\epsilon}$ for maximum throughput by equating the first derivative of $\Upsilon$ to zero. 
\begin{corollary}
The optimum explore/exploit duty cycle ($\tilde{\epsilon}$) for maximum throughput is given by
 \begin{align}
 \label{eq:corr1}
 \tilde{\epsilon} = \frac{\sqrt{a^2+4a}-a}{2}
 \end{align}
\end{corollary}
The above case shows that the duty cycle is a function of the SCNR of the JRC system (shown in $a$ in \eqref{eq:aconst}). Figure.\ref{fig:ResultsvsKappa_eps} shows the variation of $\Pdc$ and $\Upsilon$ with respect to $\epsilon$ for different values of $\kappa_m$.
 \begin{figure*}
    \centering
    \includegraphics[width=\linewidth]{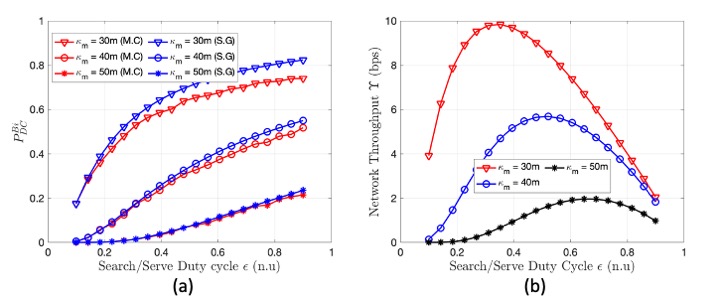}
    \caption{(a) Detection coverage ($\Pdc$) and (b) network throughput ($\Upsilon$) as a function of explore/exploit duty cycle ($\epsilon$) for parametric bistatic bistatic range ($\kappa$).}
    \label{fig:ResultsvsKappa_eps}
\end{figure*}
The view graph, Fig.\ref{fig:ResultsvsKappa_eps}a, shows that $\Pdc$ improves with increase in $\epsilon$. In other words, when we have longer search time, we can use finer beams to search for the MU and thus have a greater likelihood of detecting them. However, the same is not true for the throughput ($\Upsilon$) shown in Fig.\ref{fig:ResultsvsKappa_eps}b. An increase in $\epsilon$ initially improves the $\Upsilon$ but subsequently causes a deterioration due to the reduction in communication service time. The optimum $\tilde{\epsilon}$ in the view graph matches the estimate from the corollary \eqref{eq:corr1}.
Since the above metric is shown to be a function of $\kappa_m$, it becomes difficult for a system operator to vary $\epsilon$ according to the position of the MU. Instead, we recommend that the above tuning is carried out for the maximum bistatic range of the JRC system which is determined based on the pulse repetition frequency. The selection of the PRF is discussed in subsection \ref{subsec:PRI}. Note that in the above view graphs, the results obtained from Monte Carlo system simulations closely match the results derived from the SG based analysis.
\subsection{Monostatic Conditions}
\label{subsec:Monostatic}
The discussions so far are regarding the bistatic scenario. However, several JRC deployments are envisioned to be monostatic configurations. By setting the bistatic length $L=0$ and the bistatic angle $\beta=0$, we obtain $\Pdc$ and $\Upsilon$ for monostatic conditions. 
 \begin{figure*}
    \centering
    \includegraphics[width=\linewidth]{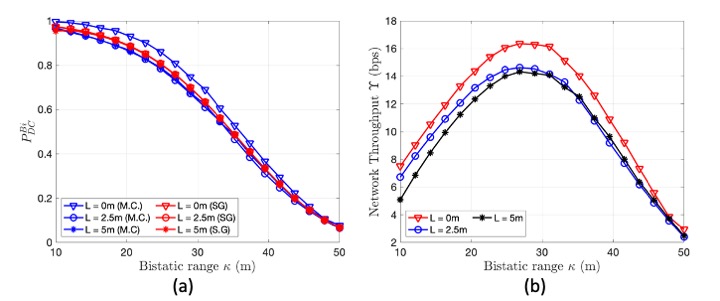}
    \caption{(a) Detection coverage ($\Pdc$) and (b) network throughput ($\Upsilon$) as a function of target bistatic range ($\kappa_m$) for parametric bistatic base length ($L$).}
    \label{fig:ResultsvsKappa_L}
\end{figure*}
Note that here, the bistatic range $\kappa_m$ can be replaced by monostatic range $r_m$ when BS and RX are co-located at the origin of the two-dimensional space. 
\begin{corollary}
The radar detection coverage metric ($\mathcal{P}_{DC}^{Mono}$) and network throughput ($\Upsilon$) for a explore/exploit duty cycle ($\epsilon$) for a monostatic radar based JRC system is given by 
\begin{align}
\label{eq:corr2}
\Upsilon = \mathcal{P}_{DC}^{Mono} \pi r_m\rho_m c\tau(1-\epsilon)D
\end{align}
where
\begin{align}
\label{eq:theorem_part4}
\mathcal{P}_{DC}^{Mono} =exp \left( \frac{-\gamma N_s r_m^4}{\longvar}+ \frac{-\gamma\rho_c c\tau r_m\rcsc}{2B_0\epsilon(\rcsm + \gamma\rcsc)}\right)
\end{align}
\end{corollary}
In Fig.\ref{fig:ResultsvsKappa_L}, we study the effect of $L$ parameter on $\Pdc$ and $\Upsilon$. Note that for all values of $L$ and $\kappa_m$ in the above study, the MU remains within the cosite region of the radar. 
The results show that the performance - in terms of both $\Pdc$ and $\Upsilon$ - does not vary significantly for change from monostatic ($L=0$) to bistatic ($L>0$) conditions. 
\subsection{SNR vs. SCR}  
\label{subsec:SNRvsSCR}
Next, we discuss the effects of noise and clutter on the performance of the JRC. As pointed out earlier, there are two terms within the $\Pdc$ in \eqref{eq:theorem_part1} and \eqref{eq:theorem_part2}. 
\begin{figure*}
    \centering
    \includegraphics[width=\linewidth]{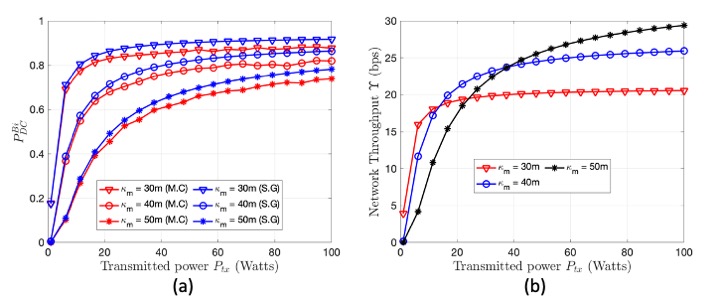}
    \caption{(a) Detection coverage ($\Pdc$) and (b) network throughput ($\Upsilon$) as a function of transmitted power ($P_{tx}$) for parametric bistatic range ($\kappa_m$).}
    \label{fig:ResultsvsPtx_Kappa}
\end{figure*}
\begin{figure*}
    \centering
    \includegraphics[width=\linewidth]{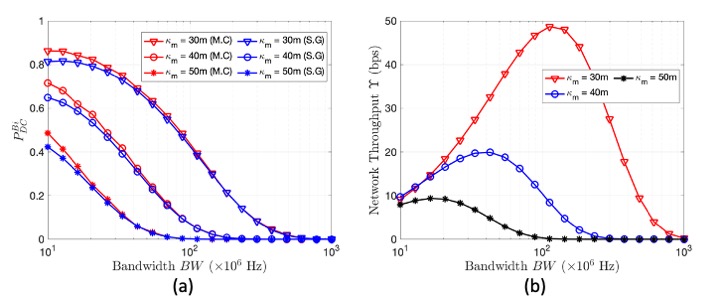}
    \caption{(a) Detection coverage ($\Pdc$) and (b) network throughput ($\Upsilon$) as a function of bandwidth ($BW$) for parametric bistatic range ($\kappa_m$).}
    \label{fig:ResultsvsBW_Kappa}
\end{figure*}
\begin{figure*}
    \centering
    \includegraphics[width=\linewidth]{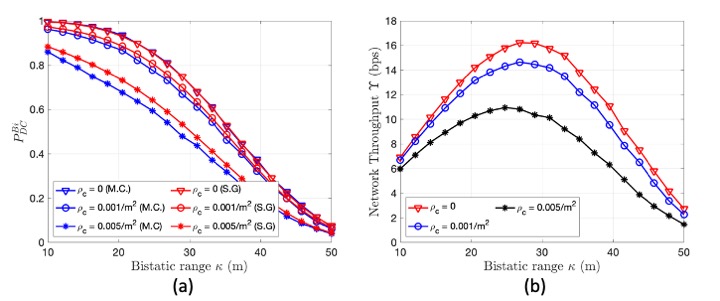}
    \caption{(a) Detection coverage ($\Pdc$) and (b) network throughput ($\Upsilon$) as a function of bistatic range ($\kappa_m$) for parametric clutter density ($\rho_c$).}
    \label{fig:ResultsvsKappa_rho}
\end{figure*}
\begin{figure*}
    \centering
    \includegraphics[width=\linewidth]{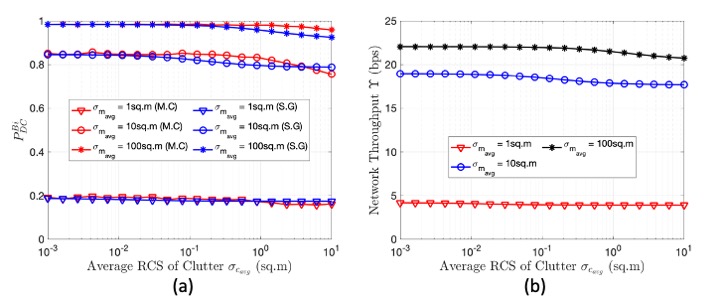}
    \caption{(a) Detection coverage ($\Pdc$) and (b) network throughput ($\Upsilon$) as a function of mean clutter RCS ($\rcsc$) for parametric mean target RCS ($\rcsm$).}
    \label{fig:ResultsvsSigmac_Sigmam}
\end{figure*}
The first term captures the effect of the signal-to-noise ratio (SNR) on the JRC performance while the second term captures the effect of the signal-to-clutter ratio (SCR). 
Figure.\ref{fig:ResultsvsPtx_Kappa} shows the effect of increasing the transmitted power $P_{tx}$ on $\Pdc$ and $\Upsilon$. The results show that $\Pdc$ and $\Upsilon$ increase initially with increase in power but subsequently, the performance saturates because the clutter returns also increase proportionately with increase in $P_{tx}$. 
On the other hand, when we consider the radar bandwidth which is the reciprocal of the pulse width ($BW = \frac{1}{\tau}$), we observe that there is an optimum $BW$ for maximum $\Upsilon$ in Fig.\ref{fig:ResultsvsBW_Kappa}b. This is because when $BW$ is increased, the range resolution decreases and correspondingly the clutter resolution cell size. As a result, fewer clutter scatterers contribute to the SCNR. But, on the other hand, the radar noise ($N_s = K_BT_sBW$) also increases which results in poorer quality radar links. 
\begin{corollary}
The optimum bandwidth $\tilde{BW}$ for maximum throughput $\Upsilon$ is obtained by the derivation of \eqref{eq:theorem_part2} with respect to $BW$ and is given by
\begin{align}
\label{eq:corr3}
\tilde{BW} = \left(\frac{\rho_c c \rcsc \rcsm P_{tx}G_0H_0}{\kappa_m^2 K_BT_s(\kappa_m+\sqrt{\kappa_m^2-L^2})(\rcsm+\gamma \rcsc)}\right)^{1/2}
\end{align}
\end{corollary}
The Monte Carlo results in Fig.\ref{fig:ResultsvsBW_Kappa}a show good agreement with SG results especially for higher values of BW. At low BW, the errors due to the path loss approximation between the point clutter and the MU become more evident. However, in real world scenarios, microwave/millimeter JRC systems are developed specifically for high bandwidth waveforms for obtaining fine range resolution of the MU.

Next we study the impact of clutter density and clutter RCS in Fig.\ref{fig:ResultsvsKappa_rho} and Fig.\ref{fig:ResultsvsSigmac_Sigmam}. When the clutter density is low ($\rho_c$ approaches zero), we observe that $\Pdc$ decays at the fourth power of $\kappa_m$ as shown in Fig.\ref{fig:ResultsvsKappa_rho} and the throughput is entirely a function of the SNR. For large values of $\kappa_m$, the system is dominated by the effects of clutter rather than noise. We observe that the throughput increases initially with increase in $\kappa_m$ due to the increase in number of MU within the area of interest and then subsequently the throughput falls due to the deterioration in the radar link metrics. 
\begin{table}[h]
    \caption{Radar, target and clutter parameters used in the stochastic geometry formulations and Monte Carlo simulations}
    \label{tab:my_params}
    \centering
    \begin{tabular}{c|c|c}
    \hline \hline
    Parameter & Symbol & Values\\
    \hline \hline
    Baselength & $L$ & 5m \\
    Transmitted power  & $P_{tx}$ & 1mW  \\
    Total time  & $T_{search}+T_{serve}$& 1s\\
    Dwell time & $T_{beam}$ & 5ms\\
    Pulse width & $\tau = \frac{1}{BW}$ & 1ns\\
    Noise temperature (Kelvin) & $T_s$ & 300K\\
    Gain constant & $G_0$& 1\\
    Threshold & $\gamma$ & 1\\
    \hline
    Mean clutter RCS & $\rcsc$ & 1$m^2$ \\
    Clutter density & $\rho_c$ & 0.01$/m^2$\\
    \hline
    Mean MU RCS & $\rcsm$ & 1$m^2$ \\
    \hline \hline
    \end{tabular}
\end{table}
The effect of $\rcsc$ is less significant on $\Pdc$ and $\Upsilon$ as both curves are flat in Fig.\ref{fig:ResultsvsSigmac_Sigmam}a and b. On the other hand, the performances are far more sensitive to $\rcsm$.
\subsection{Pulse repetition interval}
\label{subsec:PRI}
The maximum two-way unambiguous range of a radar, $R_{max} = (R^{tx}_m+R^{rx}_m)_{max}$, is equal to $cT_{PRI}$. Through the intersection of the ellipse defined for a uniform $R_{max}$ and the Cassini oval of constant $\kappa_{m}$, the two terms are related through 
\begin{align}
\label{eq:Runambig}
    R_{max}= cT_{PRI} = L^2+2\kappa_{m}^2(1+\cos\beta).
\end{align}
Note that in the above expression, the bistatic range changes for the parameter $\beta$.
The maximum value that $\cos\beta$ can take is 1. Hence, for a given radar's $T_{PRI}$
\begin{align}
    \kappa^{max}= \frac{1}{2}(c^2T^2_{PRI}-L^2)^{1/2}.
\end{align}
If we assume that at this range $\kappa_{max} >> L$, then $\Pdc(\kappa_{max})$ is given by
\begin{align}
    \Pdc(\kappa_{max}) = exp \left( \frac{-\gamma N_s (c^2T_{PRI}^2-L^2)^2}{16\longvar}+ \frac{-\gamma\rho_c c\tau\rcsc (c^2T_{PRI}^2-L^2)^{1/2}}{4B_0\epsilon(\rcsm + \gamma\rcsc)}\right), 
\end{align}
and the throughput is given by
\begin{align}
    \Upsilon(\kappa_{max}) = \Pdc(\kappa_{max})\frac{\pi}{2}(c^2T_{PRI}^2-L^2)^{1/2} \rho_m c\tau (1-\epsilon)D.
\end{align}
In the above throughput expression, it is evident that if the $T_{PRI}$ is larger, the radar detection performance deteriorates. However, a larger number of MU are included in the region-of-interest due to which there are some gains in the throughput. We assume that if the $R_{max}$ is high enough to ignore the effects of $L$, the radar operates under clutter limited conditions, and the throughput is a function of $T_{PRI}$, as given in 
\begin{align}
\label{eq:Tpri}
    \Upsilon(T_{PRI}) = exp\left(-\frac{\gamma \rho_c \rcsc c^2\tau T_{PRI}}{4B_0\epsilon(\rcsm+\gamma\rcsc)}\right)\frac{\pi}{2}c^2\tau T_{PRI}\rho_m (1-\epsilon)D.
\end{align}
\begin{corollary}
Accordingly, the optimum pulse repetition interval, $\tilde{T}_{PRI}$, can be estimated for maximum throughput as
\begin{align}
\label{eq:corr4}
    \tilde{T}_{PRI} = \frac{4B_0\epsilon(\rcsm+\gamma\rcsc)}{\gamma\rho_c \rcsc c^2\tau}.
\end{align}
\end{corollary}
The above expression shows that higher $\epsilon$ (resulting in narrow beams) and shorter pulse duration (smaller $\tau$) will allow for a longer pulse repetition interval and unambiguous range due to improvement in the link metrics.
\subsection{Meta Distribution of SCNR in a Bistatic Radar}
Although the $\Pdc$ is a useful metric for tuning radar parameters, it only provides an average view of the network across all possible network realizations of the underlying point process. This inhibits derivation of link-level reliability of the radar detection performance. 
In this regard, the meta-distribution, i.e., the distribution of the radar $\Pdc$ conditioned on a realization of $\Phi$ provides a framework to study the same. 
\begin{align}
\PdcPhi &= \mathcal{P}(\mathbf{SCNR}(\kappa_m)\geq \gamma \vert \Phi)=\mathcal{P}\left(\mathbf{\sigma_m}\geq \Su\frac{\gamma\mathbf{\sigma_c}\kappa_m^4}{\kappa_c^4}+\frac{\gamma N_s \kappa_m^4}{P_{tx}G_0B_0\epsilon H_0} \bigg| \Phi\right), \\
& = \exp\left(- \frac{\gamma N_s \kappa_m^4}{\rcsm P_{tx} G_0 B_0 \epsilon H_0}\right) \left(\Pro \left(\frac{\gamma \rcsc (R^{tx}_c)^{-2} (R^{rx}_c)^{-2} \kappa_m^4}{\rcsm + \gamma \rcsc (R^{tx}_c)^{-2} (R^{rx}_c)^{-2} \kappa_m^4 }\right)\right).
\end{align}
For a point clutter located at a distance, $R^{tx}_c$, from the transmitter at an angle $\theta^{tx}_c$, we have $ (R^{rx}_c)^2= (R^{tx}_c)^2 + L^2 + 2 R^{tx}_c L \cos (\theta^{tx}_c)$. The direct evaluation of the exact distribution of $\PdcPhi$ is challenging. Thus, we take an indirect approach to evaluate it through the calculation of its moments. 
In particular, the $b$-th moment of $\PdcPhi$ is given by:
\begin{align}
    &M_b = \mathbb{E}\left[ \mathcal{T}(b,\kappa_m) \left(\Pro \left(\frac{\gamma \rcsc (R^{tx}_c)^{-2} (R^{rx}_c)^{-2} \kappa_m^4}{\rcsm + \gamma \rcsc (R^{tx}_c)^{-2} (R^{rx}_c)^{-2} \kappa_m^4 }\right)\right)^b \right] \nonumber \\
    & =  \mathcal{T}(b,m) \mathbb{E}\left[\left(\Pro \left(\frac{\gamma \rcsc (R^{tx}_c)^{-2} (R^{rx}_c)^{-2} \kappa_m^4}{\rcsm + \gamma \rcsc (R^{tx}_c)^{-2} (R^{rx}_c)^{-2} \kappa_m^4 }\right)\right)^b\right] \nonumber \\
    & = \frac{1}{2\pi}\mathcal{T}(b,m) \int\limits_{0}^{2\pi}\exp\left(-\rho_c \int\limits_{\theta^{tx}_m - \frac{\Delta \theta_{tx}}{2}}^{\theta^{tx}_m + \frac{\Delta \theta_{tx}}{2}}\int\limits_{R^{tx} - \frac{\delta r}{2}}^{R^{tx} + \frac{\delta r}{2}} 1 - \left(\frac{\gamma \rcsc y^{-2} y_r^{-2} \kappa_m^4}{\rcsm + \gamma \rcsc y^{-2} y_{r}^{-2} \kappa_m^4 }\right)^b y dy d\theta^{tx}_{c}\right) d\theta_m\nonumber \\
    & = \frac{1}{2\pi}\mathcal{T}(b,m) \int\limits_{0}^{2\pi} \exp\left(-\rho_c \sum_{k=1}^b \binom{b}{k}  \int\limits_{\theta^{tx}_m - \frac{\Delta \theta_{tx}}{2}}^{\theta^{tx}_m + \frac{\Delta \theta_{tx}}{2}}\int\limits_{R^{tx} - \frac{\delta r}{2}}^{R^{tx} + \frac{\delta r}{2}}\left(-\frac{\gamma \rcsc y^{-2} y_{r}^{-2} \kappa_m^4}{\rcsm + \gamma \rcsc y^{-2} y_r^{-2} \kappa_m^4 }\right)^k y dy d\theta^{tx}_{c}  \right)d\theta_m,
\end{align}
where, $\mathcal{T}(b,m) = \exp\left(- \frac{\gamma b N_s \kappa_m^4}{\rcsm P_{tx} G_0 B_0 \epsilon H_0}\right)$, $y_r = (y^2 + L^2 - 2 y L \cos(\theta^{tx}_{c}))^{\frac{1}{2}}$. 
Now, for a large bandwidth, the range-resolution cell is relatively small, and hence, with the path loss approximation $\sqrt{R^{tx}_cR^{rx}_c} = \kappa_m$ for all clutter points within the cell, we have:
\begin{align}
    M_b & = \exp\left(- \frac{\gamma b N_s \kappa_m^4}{\rcsm P_{tx} G_0 B_0 \epsilon H_0}\right) \mathbb{E}_n \left[\left(\frac{\rcsm}{ \rcsm + \gamma\rcsc }\right)^{nb}\right] \nonumber \\
    & = \exp\left(- \frac{\gamma b N_s \kappa_m^4}{\rcsm P_{tx} G_0 B_0 \epsilon H_0}\right) \exp\left(\rho_c A_c(\kappa_m) \left(\left(\frac{\rcsm}{ \rcsm + \gamma\rcsc }\right)^{b} - 1\right) \right)
\end{align}
We note here that with the path loss approximation, only the number of clutter points (and not their locations) inside the range resolution cell $n$ impacts the moment.
Then, the complementary CDF of the conditional $\PdcPhi$ can be evaluated using the Gil-Pelaez inversion theorem as:
\begin{align}
    F_{\PdcPhi}(z) = \frac{1}{2} - \frac{1}{\pi} \int_0^{\infty} \frac{\mathcal{I}\left(\exp\left(-ju \log(z)\right)\right) M_{ju}}{u} du 
\end{align}
where, $j = \sqrt{-1}$ and  $M_{ju}(\cdot)$ is the $ju$-th moment of $\PdcPhi$.

In Fig.~\ref{fig:bistatic_approx} we see the impact of the path loss approximation of the clutter points on the meta-distribution of the SCNR. 
\begin{figure*}
    \centering
    \includegraphics[width=0.5\linewidth]{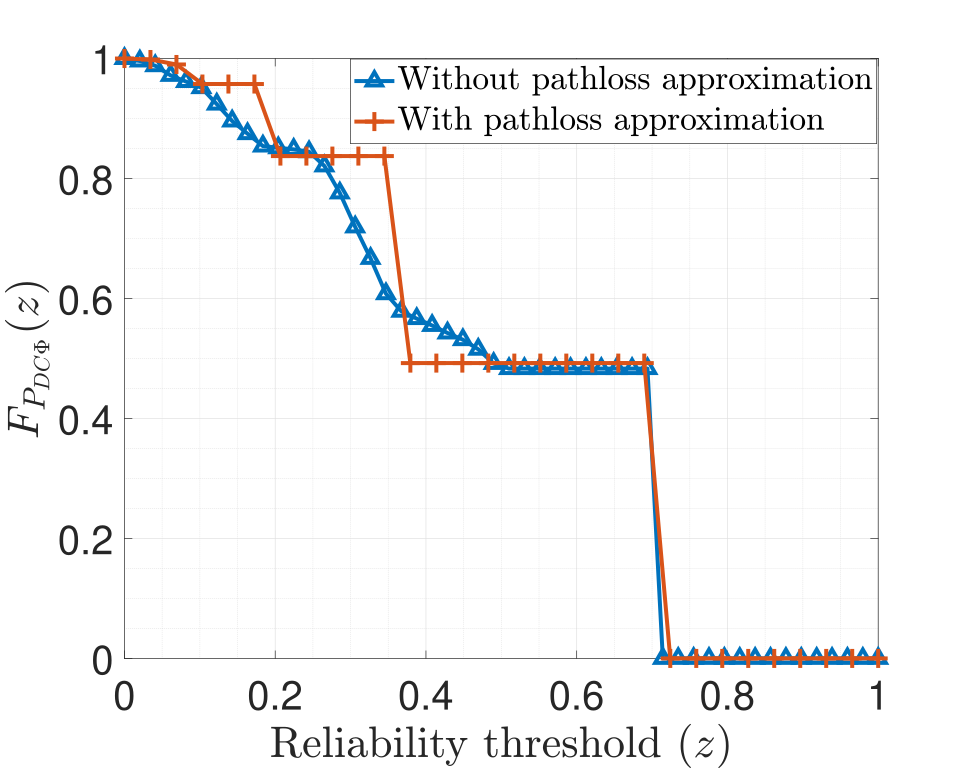}
    \caption{Comparison of the meta-distribution of the SCNR with and without the path loss approximation of the clutter points. Here $\epsilon = 0.5$.}
    \label{fig:bistatic_approx}
\end{figure*}
In particular, we see that since with the path loss approximation, the meta-distribution depends only on the number of clutter points within the range resolution cell, the corresponding plot has a stepped behaviour, where each step corresponds to a certain number of clutter points. On the contrary, the plot without the path loss approximation takes into account the relative randomness in the locations of the clutter points within the range resolution cell. For a given $\kappa_m$, the path loss approximation may result in either an overestimation or an underestimation of the actual meta-distribution. However, such an analysis is out of scope of the current work and will be investigated in a future work.
\begin{figure*}
    \centering
    \includegraphics[width=0.5\linewidth]{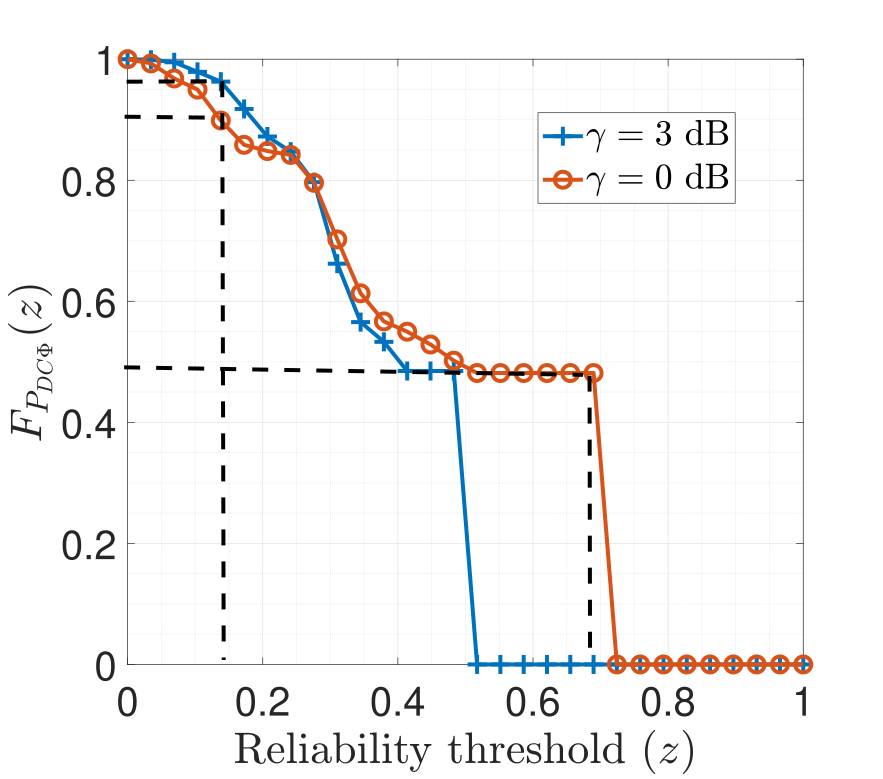}
    \caption{Meta distribution of the SCNR for different SCNR thresholds.}
    \label{fig:meta_bistatic}
\end{figure*}
In Fig.~\ref{fig:meta_bistatic} we plot the meta-distribution of the SCNR for different SCNR thresholds. This represents, qualitatively, a fine-grained analysis of the radar detection. For a given $\gamma$ the meta-distribution evaluated at a reliability threshold $z$ represents the $\mathbb{P}\left(\mathbf{SCNR}(\kappa_m)\geq \gamma \vert \Phi \geq z\right)$. Thus, the $y-$axis of the plot represents the minimum fraction of users which are detected with a reliability threshold given in the $x-$axis. For example, when the radar detection threshold is set at $\gamma = 0$ dB, we observe that about half ($F_{P_{DC\Phi}}(z) = 0.5$) of the targets are detected with a reliability of at least 70\% (i.e., $z = 0.7$), while virtually no targets ($F_{P_{DC\Phi}}(z) = 0$) are detected with a reliability of 70\% when the detection threshold is set at $\gamma = 3$ dB. On the lower reliability regime, interestingly, we observe that with $\gamma = 3$ dB, more than 95\% of the targets ($F_{P_{DC\Phi}}(z) = 0.95$) are detected with a reliability of at least 15\% (i.e., with $z = 0.15$) while the same for $\gamma = 0$ dB is lower (about 90\%). This also indicates that for a lower SCNR threshold, not only the detection probability $\Pdc$ is higher, but also guaranteeing higher reliability for individual links is more likely. Remarkably, we observe that regardless of the value of $\Pdc$, none of the targets can be guaranteed to be detected beyond 70\% ($z = 0.7$) reliability, and to achieve that, additional radar transceivers must be deployed.
\section{Conclusions}
We have provided an SG based analytical framework to provide system level planning insights into how radar based localization can enhance communication throughput of a JRC system. The key advantage of this framework is that it accounts for the significant variations in the radar, target and clutter conditions that may be encountered in actual deployments without requiring laborious system level simulations or measurement data collection. Specifically, we provide a theorem to optimize JRC system parameters such as the explore/exploit duty cycle, the transmitted power, bandwidth and pulse repetition interval for maximizing the network throughput. The results are presented for generalized bistatic radar scenarios from which the monostatic results are derived through limiting conditions. We also provide a study on the meta-distribution of the radar detection metric which provides the key insight that none of the mobile users can be reliably detected beyond 70\% of the time with a single JRC configuration. Our results are validated with Monte Carlo simulations. 

The analysis is based on some assumptions: First, we have assumed that the mobile users/radar targets fall in the cosite region of the bistatic radar when the bistatic range is greater than twice the baselength. This assumption is satisfied in several JRC applications such as indoor localization using WiFi/WLAN devices and in radar enhanced vehicular communications. However, the assumption does not hold for GNSS based bistatic radar remote sensing where the transmitter is the satellite while the receiver is mounted close to the earth. Second, we have considered short range line-of-sight links in our study which is applicable to mmWave JRC implementations. However, real world deployments encounter blockages that must be accounted for from a JRC system design perspective. Similarly, the radar will receive returns from sidelobes along with the main lobes which has not been considered in our work. Finally, in our throughput analysis, we have assumed that all the mobile users have uniform data rates that can be supported. In real world conditions, the requirements from individual users will differ and there may be system constraints on the maximum resource utilization. Hence, further analysis on network throughput based on system resource constraints would be beneficial for fine tuning JRC system parameters and would form the basis of future studies. 
\bibliographystyle{IEEEtran}
\bibliography{main}
\end{document}